\RequirePackage{amsmath}
\documentclass[runningheads]{llncs}
\usepackage[letterpaper,hmargin=1.455in,vmargin=1.3in]{geometry}
\usepackage{amssymb,amstext}
\usepackage[T1]{fontenc}
\usepackage{slantsc}
\usepackage{graphicx}

\usepackage{anyfontsize}

\usepackage{hyperref}
\usepackage{physics}
\usepackage{mathtools}
\usepackage{tikz}
\usepackage{float}
\usetikzlibrary{arrows, positioning, shadows}
\usepackage[ruled,vlined]{algorithm2e}

\usepackage{xspace}

\usepackage{tabularx} 
\usepackage{makecell}

\newcommand{\Z}{\mathbb Z}

\newcommand{\F}{\mathbb F}
\newcommand{\C}{\mathbb C}
\newcommand{\Q}{\mathbb  Q}
\newcommand{\p}{\mathbb P}

\numberwithin{equation}{section}
\usepackage[new]{old-arrows}

\usepackage{color}

\begin{document}
%
\title{Arithmetic circuit lower bounds from sumset expansion}
\titlerunning{Arithmetic circuit lower bounds from sumset expansion}
%
\author{Anand Kumar Narayanan} 
%
%
\authorrunning{Anand Kumar Narayanan}
%
\institute{\email{anandkumar.n@gmail.com}}
%
%
\maketitle              
\begin{abstract}
Raz proposed a program to prove arithmetic circuit lower bounds through the explicit construction of \textit{elusive functions} \cite{raz}. 
These are polynomial maps from a low dimensional space to a high dimensional ambient space whose image is contained in no subvariety of low complexity. 
Here, complexity is prescribed in terms of the dimension and degree of parametric maps into the ambient space defining the subvariety.   
Elusive functions are abundant: finding explicit ones with parameters typical of generic polynomial maps implies Valiant's hypothesis that VP\(\neq\)VNP.
But no such construction is known. 
Raz devised elusive functions with weaker parameters to derive explicit degree \(d\) polynomials in \(n\) variables requiring superlinear circuit size at depth \(d=o(\log n)\).\\ 

We present a new method to analyse and construct elusive functions, with coordinate maps restricted to monomials. To prove elusiveness, we identify a hitting set of points, each a tuple of roots of unity coupled based on the exponents of the monomial maps.
Using Chebotarev's theorem on roots of unity, we show that for every low complexity subvariety, the function evaluated at some point in the hitting set eludes it. 
For this strategy to work, it suffices that the iterated sumset of a certain set of numbers (derived from the exponents) expands exponentially.
We thus reduce open explicit construction problems in elusive functions (entailing complicated polynomial constraints) to purely additive combinatorial ones, whose resolutions imply as yet unknown lower bounds. \\

Informed by iterated sumset expansion, we devise new elusive functions. 
We construct explicit elusive curves of exponential degree, resolving an open problem posed by Garg, Makam, Oliveira, and Wigderson \cite{gmow} as a testament to the difficulty of elusiveness proofs. 
The exponential degree obstructs the deduction of lower bounds, but even lowering it to  subexponential would imply VP\(\neq\)VNP. 
We improve Raz's superlinear bound quadratically (with circuit size to input size ratio as the metric) below \(o(\log n/\log\log n)\) depths. The degree of our explicit polynomial is comparable to Raz's at the deep end of \(o(\log n/\log\log n)\), but much worse at shallower depths. Our methods can also prove rigidity of symbolic matrices, high-rank of symbolic tensors, and more generally find points outside algebraic natural proofs.  
As an example, we present semi-explicit high border rank tensors over smaller degree number fields than previously known.
%
%
\end{abstract}
\section{Introduction}\label{section-introduction}
\subsection{Arithmetic circuits}
We consider the model of arithmetic circuits as in \cite[\S~1.1]{raz} and refer to  \cite{vzg} for a broader discussion. We only consider arithmetic circuits over the complex field. Such a circuit is built on a directed acyclic graph. The vertices of in-degree zero are called leaves and are labeled with either a variable or the constant one. The other vertices are either labeled as a sum gate or a product gate, and perform the corresponding operation on its incoming edges. We do not restrict the fan-in (in-degree) to the gates. Each edge is labeled with a constant and the number of edges is the size of the arithmetic circuit. The length of the longest directed path is the depth. The vertices of out-degree zero are called as the outputs. Each output is a polynomial in the input variables. A fundamental question is to explicitly construct polynomials that require  arithmetic circuits of a certain size and depth. Explicit  means poly-definable/poly(n)-definable in the sense of Valiant \cite{val-pdef,val-red}. Informally, a polynomial is poly(n)-definable if there is a polynomial (in the number of inputs/variables \(n\)) sized circuit to uniformly compute its coefficients (see definition \ref{definition-polydefinable} or \cite[\S~1.5]{raz}, for alternate characterizations see \cite[Thm. 4.2]{vzg}).   
The central problem in the complexity of arithmetic circuits is Valiant's hypothesis VP\(\ne\)VNP, that there are poly(n)-definable polynomials without polynomial sized circuits \cite{val-pdef,val-red}. 
Raz devised an enticing program to prove VP\(\ne\)VNP, through polynomial maps that he called elusive functions \cite{raz}. He showed how to construct from explicit elusive functions (see definitions \ref{definition-elusive} and \ref{definition-explicit-elusive}) with strong enough parameters, poly(n)-definable polynomials that require superpolynomial sized circuits. Therefore, explicit elusive functions imply VP\(\ne\)VNP. 
Elusive functions are abundant, with a generic polynomial map being elusive with strong parameters. Yet, there are few known techniques to prove elusiveness, making explicit construction a notorious ``hay in a haystack'' problem. We present new techniques to prove elusiveness of monomial functions. Informed by the proof technique, we construct new monomial elusive functions. These are not strong enough for superpolynomial lower bounds: but either improve marginally on the state of the art polynomial lower bounds at low depths or reduce lower bound problems to additive combinatorics. 
\subsection{Elusive curves and Chebotarev's other theorem}
    Call a function \(\Gamma:\C^s\longrightarrow\C^m\) a degree-r map if its image \(\Gamma(\C^s)\) has a parametric equation of the form 
     \((z_1,z_2,\ldots,z_s)\longmapsto \left(g_i(z_1,z_2,\ldots,z_s)\right)_{i\in[m]},\) 
    where each \(g_i(z_1,z_2,\ldots,z_s)\) is a polynomial in \(\C[z_1,z_2,\ldots,z_s]\) of degree at most \(r\). When \(s\) (in comparison to \(m\)) and \(r\) are small, we informally call such \(\Gamma\) low complexity maps. 
\begin{definition}\label{definition-elusive}
    A function \(f:\C^n\longrightarrow\C^m\) is defined to be \textbf{\((s,r)\)-elusive} if \(f(\C^n) \nsubseteq \Gamma(\C^s)\) for every degree-r map \(\Gamma:\C^s \longrightarrow \C^m\). 
\end{definition}    
Raz noted that \(z\longmapsto \left(z,z^2,z^3,\ldots,z^{m}\right)\)
is \((m-1,1)\)-elusive (the image is called the moment curve), since the images of \(m\) distinct non zero values of \(z\) do not all lie on a hyperplane, as evident from the non-singularity of the Vandermonde matrix associated with those \(m\) points on the curve. He then asked for an explicit  construction of an \((m-1,2)\)-elusive curve, proving that such elusive curves with the coordinates being degree-\(2^{o(m)}\) maps implies VP\(\neq\)VNP. The demands of this explicit construction can be eased by lowering the dimension of the low complexity maps and increasing the dimension of the elusive functions, for instance \((m^{9/10},2)\)-elusive functions with \(m=n^{\omega(1)}\) suffice. 
Garg, Makam, Oliveira, and Wigderson note how little progress has been made since on the construction of elusive functions: ``\textit{This beautiful avenue to proving superpolynomial lower bounds is a great challenge to our techniques, and no progress we know of was made since that paper came out. Here we will attempt to handle a very toy version of it using our numeric to symbolic transfer}'' \cite[\S~9]{gmow}. They asked if 
\(z\longmapsto (z,z^3,z^9,\ldots,z^{3^{m-1}})\) is \((m-1,2)\)-elusive (or even \((m^{9/10},2)\)-elusive), in hopes that allowing exponential degree maps would help prove elusiveness. But their numeric to symbolic transfer of algebraic relations to algebraic functions could only prove a much weaker notion of elusiveness, with severe restrictions on the low complexity maps that need eluding \cite[Prop. 9.8]{gmow}. 
We give elementary proofs that such curves (with even smaller powers) are \((m^{\lfloor9/10\rfloor},2)\)-elusive, deduced from the additive structure of the exponents. The curve \(z\longmapsto (z,z^2,z^4,\ldots,z^{2^{m-1}})\) is not \((m-1,2)\)-elusive, being contained in the hypersurface \(\Gamma(\C^{m-1})\) with a degree-\(2\) map \((z_1,z_2,\ldots,z_{m-1})\xmapsto{\Gamma} (z_1,z_1^2,z_2\ldots,z_{m-1})\) (see remark \ref{remark-low-degree-annihilators}). 
\begin{theorem}\label{theorem-elusive-curve-intro}
    For large enough \(m\), \(z\longmapsto (z,z^2,z^4,\ldots,z^{2^{m-1}})\) is \((\lfloor m^{9/10}\rfloor,2)\)-elusive.
\end{theorem}
Any significant lowering of the degree in theorem \ref{theorem-elusive-curve-intro}, such as to a subexponential \(2^{m^{o(1)}}\) implies VP\(\ne\)VNP. 
Two motifs from the proof of theorem \ref{theorem-elusive-curve-intro} will be abstracted and used repeatedly. First, the set of exponents \(A=\{1,2,4,\ldots,2^{m-1}\}\) defining the curve has the property that the iterated sumsets \(A+A+\ldots+A\) grow exponentially in size, evident from the uniqueness of binary expansion. 
Our broader plan is to study and construct elusive monomial functions \(f:\C^n\longrightarrow\C^m\) where each of the \(m\) coordinate maps is a monomial. The expansion of the iterated sumsets of certain sets crafted from the monomial exponents will play a pivotal role in the analysis of such monomial maps. Second, we will look for the images of points with \(p\)-th roots of unity (for a large prime \(p\)) coordinates to elude low complexity maps. In the proof of theorem \ref{theorem-elusive-curve-intro}, simply the images of a \(p\)-th root of unity and its conjugates are enough. 
Later (for lemma \ref{lemma-elusive}, theorems \ref{theorem-reduction-intro} and \ref{theorem-rigidity-intro}), images of tuples of \(p\)-th roots of unity that are intricately coupled by a relation derived from the monomial exponents are needed. To prove that such images are outside low complexity maps, we invoke Chebotarev's theorem on roots on unity, which asserts that every minor of a prime order discrete Fourier matrix is non zero (see \cite{sl} for Chebotarev's \(p\)-adic proof). 
That is, 
\begin{center}
    \textit{ For every prime \(p\), every primitive \(p\)-th root of unity \(\zeta\in \C\), and every two non empty subsets \(A,B \subseteq \{0,1,\ldots,p-1\}\)} of the same size, \(\det\left(\left(\zeta^{ab}\right)_{a \in A,b\in B}\right)\neq 0\).
\end{center}
In essence, we recast the role played by the structure theorem of Vandermonde matrices with Chebotarev's theorem, to elude higher complexity maps. 
The analogous statement of Chebotatev's theorem is false in positive characteristic, hence our claims are only on arithmetic circuits over \(\C\).  The bounds also hold for arithmetic circuits over subfields of \(\C\), such as \(\Q\) or other number fields, since our elusive functions will have zero/one coefficients and a function being elusive over a field implies bounds over all subfields \cite{raz}. See remark \ref{remark-positive-characteristic} on how to extend our results to fields of very large positive characteristic.   
\subsection{Expansion of iterated sumsets implies elusiveness}\label{subsection-intro-expansion-elusive}
We next describe the connection between elusive monomial maps and sumset expansion. 
    For a finite non empty set \(S \subset \Z_{\ge 0}\) and positive integers \(k,\ell\), define the \(k\)-fold sumset
    \[kS:=\underbrace{S+S+\ldots+S}_{k \text{ copies}}\ \ = \bigcup_{(s_1,s_2,\ldots,s_k)\in S^k}(s_1+s_2+\ldots+s_k) \] 
    and let 
    \(\ell^\le S:= \bigcup_{k=1}^\ell kS.\) 
Since the empty summation is not allowed, \(0\in \ell^\le S\) only if \(0 \in S\). For a positive integer \(a\), let \([a]\) denote \(\{1,2,\ldots,a\}\). To prove elusiveness of a function \(f:\C^n\longrightarrow\C^m\) mapping 
\[(z_1, z_2,\ldots, z_n) \longmapsto \left(\prod_{j\in [n]}z_j^{\alpha_{i,j}}\right)_{i\in [m]}, \]
we establish purely additive sufficient conditions on its  \(m \times n\) \textit{exponent matrix} \(A=(\alpha_{i,j})\). Consider an auxiliary \textit{big vector} \(\beta=(\beta_1,\beta_2,\ldots,\beta_n)\) of positive numbers and let 
\begin{equation}\label{equation-exponent-vector-intro}
    A^\beta:=\left(\sum_{j\in[n]}\alpha_{i,j}\beta_j\right)_{i\in [m]}
\end{equation}
denote its \textit{evaluation vector}, which abusing notation we will also look at as a set. 
To prove elusiveness, for every low complexity map \(\Gamma\) we wish to demonstrate a polynomial \(h_\Gamma\) in the coordinate ring \(C[x_1,x_2,\ldots,x_m]\) of \(\C^m\) that vanishes at \(\Gamma(\C^s)\), but not on \(f(\C^n)\). To this end, we seek polynomials \(h_\Gamma\) of degree at most a carefully chosen parameter \(\ell\). 
First, we construct a set of distinct monomials \(M \subset \C[x_1,x_2,\ldots,x_m]\) indexed by elements in the sumset \(\ell^\le A^\beta\), where each monomial is of degree at most \(\ell\). Their \(\C\)-linear span \(\text{Span}(M)\) is \(|\ell^\le A^\beta|\) dimensional. For an arbitrary degree-\(r\) map \(\Gamma:\C^s\longrightarrow\C^m\), consider the \(\C\)-linear composition map \(\iota:\text{Span}(M)\longrightarrow\C[z_1,z_2,\ldots,z_s]\) taking \(h\longmapsto h\circ \Gamma\). By \(h\circ \Gamma\), we mean substituting for \(x_i\) in \(h\), the \(i\)-th coordinate polynomial of \(\Gamma\) (which is, of degree at most of \(r\) in the variables \(z_1,z_2,\ldots,z_s\)). The image of \(\iota\) is therefore contained in the subspace of polynomials in \(z_1,z_2,\ldots,z_s\) of degree at most \(\ell r\), an \(\binom{s+\ell r}{\ell r}\)-dimensional space. Therefore, if 
\[|\ell^\le A^\beta| > \binom{s+\ell r}{\ell r},\]
then \(\ker(\iota)\) is non trivial and contains a non zero polynomial \(h_\Gamma\) that vanishes at \(\Gamma(\C^s)\). 
We will then use Chebotarev's theorem to show that for \(\zeta\) a primitive \(p\)-th root of unity with \(p\) larger than the maximum in \(\ell^\le A^\beta\), at least one of the points \(\left\{f\left(\zeta^{b\beta_1},\zeta^{b\beta_2}, \ldots,\zeta^{b\beta_n}\right),b\in[p]\right\}\) is not a root of \(h_\Gamma\). Therefore, \(h_\Gamma\) does not vanish on \(f(\C^n)\), proving \(f\) is \((s,r)\)-elusive. 
In retrospect, we can view \(\left\{\left(\zeta^{b\beta_1},\zeta^{b\beta_2}, \ldots,\zeta^{b\beta_n}\right),b\in[p]\right\}\) as a hitting set, with each point coupled by the big vector \(\beta\). For every low complexity map \(\Gamma\), the evaluation of \(f\) at (at least) one of the points in the hitting set eludes \(\Gamma\). 
We warn that the elusive functions we construct can satisfy low degree polynomials in the coordinate ring (see remark \ref{remark-low-degree-annihilators}), but this is compatible with being elusive (for \(s\ll m-1\)).   
Using our proof technique for elusiveness, we can write down additive combinatorial problems corresponding to the elusive function construction problems.  
For instance, the following theorem is the analogue of \cite[Cor. 3.9]{raz}.
\begin{theorem}\label{theorem-reduction-intro}
        Let \(r,n,s\) be positive integers with \(1\le r \le n\) and \(m:=n\binom{n+r-1}{r}\), where \(s = n^{\omega(1)}\) is super polynomial in \(n\). If there exists an explicit  \((m,n,s/r,4^s)\)-expander \(A =(\alpha_{i,j})_{(i,j)\in[m]\times[n]}\)  with exponents bounded as \(\alpha_{(i,j)}=2^{n^{O(1)}}\), then VP\(\neq\)VNP.    
\end{theorem}
Here, an \(m\times n\) exponent matrix \(A\) is called an \((m,n,\ell,C)\)-expander if there is a big vector \(\beta\in \Z_{>0}^n\) such that \(|\ell^\le A^\beta|\ge C\) (see definitions \ref{definition-expanding-exponent} and \ref{definition-explicit} for details). 
Similarly, other elusive function construction problems in \cite{raz} can also be translated to combinatorial problems. 
We do not know how to construct such expanding exponential matrices with guarantees of superpolynomial bounds, but speculate briefly on the opportunities and difficulties. 
For simplicity, consider \(\ell A^\beta\) instead of \(\ell^\le A^\beta\) with \(\ell \approx s\). The \(\ell\)-fold expansion varies widely \(\ell|A^\beta|\le |\ell A^\beta|\le \binom{|A^\beta|+\ell}{\ell} \), from linear to exponential in \(\ell\). To land on the exponential side, it is necessary that \(A^\beta\) has numbers exponential in \(\ell\approx s\). It is sufficient if \(A^\beta\) is a pseudorandom set of such big numbers \cite{nat}. 
Entries of \(A\) can be exponential in \(n\), but not in \(s\). The only way to induce exponential in \(s\) numbers into \(A^\beta\) is through \(\beta_k\)s, which our framework allows for free without affecting any of the parameters. 
The constraint is that the number of elements in \(\beta\) is the dimension of the elusive function \(n\), a precious resource. The challenge is to design the matrix \(A\) such that for some \(\beta\), \(A^\beta\) looks pseudorandom with big numbers, without concentrating close to some \(n\) dimensional lattice. Such additive structure can prevent expansion, and in some sense is the only obstruction to expansion (see Frieman's theorem \cite{fri} or the Erd{\H{o}}s-Szemer{\'e}di \cite{es}/Erd{\H{o}}s \cite[Prob. 7]{erd} sum-product conjectures). 
We only need to show this for some \(\beta\), and it is sufficient to show for random \(\beta\).  
Therefore, we can phrase the exponent matrix  design as a pseudorandomness problem. The seed is the big vector \(\beta\) drawn uniformly from \([4^{O(s)}]\). The exponent matrix \(A\) is the linear pseudorandom generator transforming the length \(n\) big vector into the larger length \(m\) evaluation vector \(A^\beta\), whose iterated sumsets we wish expand. We do not know if such exponent matrices exist, and ask both the questions of their existence and explicit construction. 
There are \((m,n,\ell,C)\)-expanders for \(n\approx\log C\) with \(\alpha_{i,j} \in \{0,1\}\), since Cayley graphs of \(\Z/p\Z\) with \(\Omega(\log p)\) random generators are expanders \cite{hil} (think of \(p\) as the large prime in the proofs). 
We next describe explicit such expanders.  
\subsection{Superlinear bounds at sublogarithmic depth}
For parameter regimes tailored to proving superlinear lower bounds at small depth, we design new explicit exponent matrices that are \((n^2,n\log_2n,n,n^n)\)-expanders. One construction uses error correcting codes (from lossless unbalanced bipartite expanders) to build sumsets with the best possible expansion (see lemma \ref{lemma-elusive-code}). 
The other is ad-hoc, but ends up with smaller exponents (see lemma \ref{lemma-construction-main}).  
We start with a sparse \(n \times n \times n \times \log_2n\) format \(4\)-tensor. 
The tensor entries are zero/one, encoding \(n\)-ary and binary expansion of numbers designed to appear in the exponent vector \(A^\beta\) (as pictured in equation \ref{equation-design-vector}), with a fixed big vector \(\beta\) in mind. The \(4\)-tensor is then flattened to arrive at the exponent matrix. Proof of expansion follows from the uniqueness and expressiveness of binary/\(n\)-ary expansions, after substituting the big vector \(\beta\).    
We then prove in lemma \ref{lemma-elusive} that the function associated with the exponent matrix is elusive. The proof also serves as a blueprint to adapt the general strategy outlined in \S~\ref{subsection-intro-expansion-elusive} to other parameter ranges, such as in proving theorem \ref{theorem-reduction-intro} (at the end of \S~\ref{subsection-elusive}). The elusive function is then wrapped in the framework of \cite[\S~3,3 and \S~4.2]{raz} to prove our new lower bounds (see \S~\ref{subsection-lower-bounds}), which is superlinear at low depths, say \(d=o(\log n/\log\log n)\).  
\begin{theorem}\label{theorem-lower-bound-intro}
Let \(c,d>5\) be positive integers, set \(n:=2^c\).  
Every depth-\(\lfloor d/3\rfloor\) complex arithmetic circuit to compute the poly-definable polynomial  
\[ \sum_{u \in [n]} z_u \sum_{v\in [n]} y_v \prod_{(w,k)\in [n] \times [c+1]}x_{(w,k)}^{\gamma_{v,k}\delta_{(u,w)}}\]
(of degree \(c+3\) in \(n(c+3)\) variables) has size at least \((n^{1+1/d})/5,\) where \(\sum_{k=1}^{c+1}\gamma_{v,k}2^{k-1}=v\)
is the binary expansion of \(v\) and \(\delta_{(,)}\) is the Kronecker delta function.
\end{theorem}
Without restrictions on depth, lower bounds for explicit polynomials of polynomial degree are rare, with the pioneering works of Strassen \cite{str} and Baur-Strassen \cite{bs} being notable exceptions. They construct explicit polynomials  of degree \(r\) in \(n\)-variables that need circuits of size at least \(\Omega(n\log r)\) \cite[Cor. 1,2]{bs}. Better bounds are known restricted to low depths. 
At constant depths,  Limaye, Srinivasan, and Tavenas \cite{bds} prove super polynomial lower bounds for explicit polynomials! See also \cite{bcs} for improvements. Their bounds are trivial beyond depths \(d=\omega(\log\log n)\) (a conservative estimate), hence we look at the depth range \(\omega(\log\log n)\) to \(o(\log n/\log\log n)\). Here, only superlinear bounds are known, with the best due to Shoup-Smolensky \cite{ss} and Raz \cite{raz}. We next compare our result with theirs, using the ratio \(\rho\) of the circuit size to the input size (that is, the number of variables) as a measure of superlinearity. 
Another parameter to consider is the degree \(r\) of the explicit polynomials. 
In Shoup-Smolensky, evaluation points are taken as part of the input, so that their bounds are for explicit polynomials. 
\renewcommand{\arraystretch}{1.4}
\begin{center}
\begin{tabularx}{0.99\textwidth} { 
  | >{\centering\arraybackslash}X 
  | >{\centering\arraybackslash}X 
  | >{\centering\arraybackslash}X 
  | >{\centering\arraybackslash}X
  |}
 \hline
  Depth \(d\), Variables \(n\) & Shoup-Smolensky \cite{ss} & Theorem \ref{theorem-lower-bound-intro} & Raz \cite[Cor. 4.6]{raz}\\ 
 \hline
  Superlinearity \(\rho\) & \(\Omega(dn^{1/d})\) & \(n^{1/d}/(5\log n+15)\) & \(n^{1/(2d)}/(25d+10)\)\\ 
\hline
 Degree \(r\) & \(\Theta(n)\) & \(\Theta(\log n)\) & \(\Theta(d)\) \\ 
\hline
\end{tabularx}
\end{center}
Shoup-Smolensky have the best superlinearity across the depth range, but the worst (highest) degree. Raz has the worst superlinearity but the best (lowest) degree, the latter being constant at constant depth! Our results interpolate between the two. At the shallow end \(d\approx \log\log n\), our superlinearity is close to Shoup-Smolensky, both quadratically better than Raz. Our degree is exponentially better than Shoup-Smolensky, but exponentially worser than Raz. A little off the deep end \(d \ll \log/\log\log n\), our degree is close to that of Raz (which is best), but our superlinearity remains quadratically better. Shoup-Smolensky have better superlinearity but their degree is exponentially worser than ours. We must mention that Raz's bounds are the only ones applicable to constant degree polynomials, and may have been optimized towards this important \cite[\S~1.3]{raz} goal. 
\subsection{Sumset expansion and symbolic rigid matrices}
Theorem \ref{theorem-reduction-intro} asks for \((m,n,s/r,4^s)\)-expanders with the number of iterations \(s/r\) being far fewer than the number of coordinates \(m\). In \S~\ref{section-rigidity}, via matrix rigidity like problems, we extend the range of interest for sumset expansion problems with lower bound implications, and ask for expanders where the number of iterations is far greater than the number of coordinates.   
Rigid matrices are those that are far (in Hamming distance) from low rank matrices. Valiant proposed the thier explicit construction as a way to prove superlinear bounds of the size of circuits computing the corresponding linear transformation (matrix-vector product) \cite{val}.  
Volk and Kumar proved that non-rigid \(n\times n\) matrices (parametrized by bounds on rank and Hamming distance) satisfy a polynomial of degree \(n^3\) in the coordinate ring of the matrix \cite[Thm. 1.1]{vk}. They also proved a more direct characterization : \(n \times n\) matrices whose linear transformation can be computed by a linear circuit of size at most \(s\ll n^2\) satisfy a polynomial in the matrix coordinates of degree at most \(n^3\) \cite[Thm. 1.3]{vk}. Parts of their proof are more reminiscent of Raz's techniques for bounds from elusive functions/high-rank tensors \cite{raz,raz-ten}, than that of Valiant. In particular, to prove theorem \cite[Thm. 1.3]{vk}, they show that there is a universal polynomial map \(U:\C^{2s}\longrightarrow\C^{n \times n}\) of degree a small polynomial in \(n\) whose image contains all \(n \times n\) matrices whose linear transformation can be computed by a linear circuit of size at most \(s \ll n^2\). Therefore, to construct matrices whose linear transformation needs circuits of size at least \(s\), it suffices to elude this one polynomial \(U\). Pausing to compare with elusive functions: there we had to elude every low dimensional low degree map, here we have to elude one fixed low dimensional but very high degree map, higher than the number of coordinates.  
We adapt our technique outlined in \S~\ref{subsection-intro-expansion-elusive} to elude \(U\) and construct symbolic matrices whose linear transformation needs circuits of size at least \(s\), given an explicit construction of expanding exponent matrices.   
We transform this symbolic linear transformation problem by rewiring the symbols as inputs to the computation, to arrive at explicit polynomials. Thereby, we prove that an explicit construction of exponent matrices that are expanders implies superlinear circuit size lower bound for certain explicit polynomials related to the linear transformation of the matrix, notably without any bound on the depth.
\begin{theorem}\label{theorem-rigidity-intro}
        Let \(n,n',s\) be positive integers with \(n \le n'<s \le n^2\). 
        If there exists an explicit \((n^2,n',n^3,n^{8s})\)-expander \(A =(\alpha_{(u,v),j})_{(u,v)\in[n]\times[n],j\in[n']}\)  with exponents bounded as \(\alpha_{(u,v),j}\in n^{O(1)}\), then the poly-definable polynomial 
        \[\sum_{u\in[n]} z_u \sum_{v\in[n]} y_v\prod_{j\in[n']}x_j^{\alpha_{(u,v),j}}\]
        needs arithmetic circuits over the complex numbers of size at least \(s/5\).
\end{theorem}

Theorem \ref{theorem-rigidity-intro} offers an avenue to add to the rarefied list \cite{str,bs} of superlinear bounds for explicit polynomials (without any restriction on depth) of polynomial degree, by constructing explicit \((n^2,o(s),n^3,n^{8s})\)-expanders, say for \(s=\lfloor n\log\log n\rfloor\).
\subsection{Semi-explicit high rank tensors}
Finding explicit tensors of high rank is a difficult problem (see remark \ref{remark-highrank-tensor-explicit}), with implications of  unconditional arithmetic circuit lower bounds. Strassen showed that the rank of a three dimensional tensor is a lower bound on the size of arithmetic circuits computing the corresponding trilinear form \cite{str-div}. For instance, superlinear rank bounds would imply superlinear bounds for general arithmetic circuits, even for constant degree polynomials! In higher dimensions, Raz \cite{raz-ten} showed that an explicit construction of \(d\)-dimensional \(n\times n \times \ldots \times n\) tensors of tensor rank \(n^{d(1-o(1))}\) in dimension \(d\approx \log n/\log \log n\) implies every arithmetic formula for the multilinear form of the tensor (and also for the permanent, by Valiant's completeness \cite{val-completeness}) is of superpolynomial size. An arithmetic formula is an arithmetic circuit whose underlying directed acyclic graph is a tree. By an argument similar to proof of theorem \ref{theorem-rigidity-intro}, we can construct high-rank symbolic tensors, eluding the universal map describing high border rank tensors  in \cite[Lem. 4.3]{vk}. The border rank, which is at most the tensor rank, is a robust version of tensor rank (see \S~\ref{subsection-tensor-rank}). 
This would give arithmetic formula lower bounds for computing the symbolic tensor over the field \(\C(z_1,z_2,\ldots,z_{n'})\), where \(n'\) is the dimension of the elusive map. If one targets a low enough dimension \(n'=n^{o(\log d)}\), one can rewire those \(n'\) inputs and write down a theorem similar to theorem \ref{theorem-rigidity-intro}, with implications for superpolynomial formula lower bounds. We refrain from writing out the details, but give a related construction (using sumset expansion with superlinear iterates) for semi-explicit high rank tensors. Semi-explicit means the coordinates have a short description (say in a number theoretic sense as an algebraic number), but not necessarily in a complexity theoretic sense to give poly\((n)\)-definable functions. 
Specializing symbolic high-rank tensors in one variable (that is, \(n'=1\)), we get semi-explicit high-rank tensors (with height one coordinates) over exponentially (in the number of tensor coordinates) smaller degree number fields than previously known (see remark \ref{remark-highrank-tensor}). 
\begin{theorem}\label{theorem-highrank-tensor-intro}
    Let \(n,d\) be positive integers. Let \(p\ge n^{2d(n^d+1)}\) be a prime and \(\zeta\in \C\) a primitive \(p\)-th root of unity. 
    Every \(d\)-dimensional \(n\times n \times \ldots \times n\) tensor whose coordinates (disregarding ordering) form the set \(\left\{\zeta^{n^{2di}}, i\in[n^d]\right\}\)  
     has border rank at least \(n^{d-1}/2d\), for large enough \(n^d\). 
\end{theorem}
For \(d=\Theta(\log n/\log\log n)\), since the border rank is at least \(n^{d-1}/2d = n^{d(1-o(1))}\), \cite[Cor. 6]{raz-ten} implies that there is no polynomial (in the \(nd\) inputs) sized arithmetic formula to compute the (not necessarily poly(\(nd\))-definable) multilinear form of a tensor in theorem \ref{theorem-highrank-tensor-intro}.  
\subsection{Towards succinct hitting set generators against algebraic natural proofs}
In \S~\ref{section-rigidity} and \S~\ref{section-tensors} (also \cite{vk}), polynomials vanishing at complexity classes (at small sized linear circuits and low border rank tensors respectively) were critical to the lower bounds. Grochow-Kumar-Saks-Saraf \cite{gkss} and Forbes-Shpilka-Volk \cite{fsv} present abstractions that study many lower bound methods in a unified manner \cite{gkss,fsv}, calling such polynomials (that annihilate algebraic complexity classes) algebraic natural proofs (\cite[Def. 1.1]{fsv}).
Finding a hard polynomial for a complexity class amounts to finding a non vanishing point of the annihilator.
Often, an annihilator is derived from the kernel of a linear map of polynomials, counting dimensions to certify a non-trivial kernel. Our methods should help in such scenarios, to construct (or prove) succinct hitting set generators.  
The thrust of the algebraic natural proofs formalism is to identify conditional obstructions to proving significant lower bounds by such methods \cite{gkss,fsv}, analogous to the Razborov-Rudich natural proofs barrier \cite{rr}. A missing ingredient in clarifying the obstruction is the knowledge of polynomials with small circuits that are pseudorandom (unlike the Boolean case, where small Boolean cryptographically pseudorandom circuits are known). To mitigate the issue, \cite{gkss,fsv} define pseudorandom polynomial families called succinct hitting sets, whose existence would imply an algebraic natural proof barrier. Forbes-Shpilka-Volk defined more demanding objects called as succinct hitting set generators \cite[Def. 1.4]{fsv} and constructed them in  special cases. 
These share some aspects with elusive functions, such as being of small dimension and yet eluding higher dimensional subvarieties. 
A curious question is if our methods apply to the construction of succinct hitting sets (and generators), thereby  demarcating obstructions to proving arithmetic circuit lower bounds.  
\section{Explicit elusive curves of exponential degree}\label{section-elusive-curves}
\begin{proof}[of theorem \ref{theorem-elusive-curve-intro}]
    Call the curve \(z\longmapsto (z,z^2,z^4,\ldots,z^{2^{m-1}})\) in question \(f:\C\rightarrow\C^m\). Set \(s:=\lfloor m^{9/10}\rfloor\) and \(\ell:=\lfloor s/2\rfloor\). Let \(\C[x_1,x_2,\ldots,x_m]\) denote the coordinate ring of \(\C^{m}\) and \(\mathcal{M}\) its \(\C\)-linear subspace generated by the basis  
    \[\left\{x_1^{a_1}x_2^{a_2}\ldots x_m^{a_m} \middle| \ a=(a_1,a_2,\ldots,a_m)\in \{0,1\}^m, d_H(a)= \ell\right\}\]
    of multilinear monomials of degree exactly \(\ell\), indexed by bit strings \(a\) of Hamming weight \(d_H(a)\) exactly \(\ell\). The strategy is to construct for every degree-\(2\) map \( \Gamma:\C^s \longrightarrow \C^{m}\), a polynomial \(h_\Gamma \in \mathcal{M}\) that vanishes on \(\Gamma(\C^s)\), but not on \(f(\C)\). To this end, consider an arbitrary such degree-\(2\) map 
    \(\Gamma:\C^s \longrightarrow \C^m\) taking 
    \[(z_1,z_2,\ldots,z_s)\longmapsto (g_i(z_1,z_2,\ldots,z_s))_{i\in[m]}.\]
Let \(\C[z_1,z_2,\ldots,z_s]_{\le 2\ell}\) denote polynomials in \(z_1,z_2,\ldots,z_s\) of degree at most \(2\ell\). 
The substitutions \(x_i\longmapsto g_i(z_1,z_2,\ldots,z_s)\) at the \(i\in[m]\) coordinates induce the \(\C\)-linear map 
\begin{align*}
   \iota_\Gamma: \mathcal{M} &\longrightarrow \C[z_1,z_2,\ldots,z_s]_{\le 2\ell} \\
   x_1^{a_1}x_2^{a_2}\ldots x_m^{a_m} &\longmapsto \prod_{i\in[m]}g_i(z_1,z_2,\ldots,z_s)^{a_i},
\end{align*}
which has a kernel \(\ker(\iota_\Gamma)\) of dimension at least one (for large enough \(m\)), since
\[\dim(\mathcal{M}) = \binom{m}{\ell} > \left(\frac{m}{\ell}\right)^\ell \approx 2^{(s\log_2m)/10}  \gg  4^s > \binom{s+2\ell}{2\ell}= \dim(\C[z_1,z_2,\ldots,z_s]_{\le 2\ell}). \]
Hence, there exists a non zero  
\[h_\Gamma(x_1,x_2,\ldots,x_m) = \sum_{\substack{a \in \{0,1\}^m
    \\  d_H(a) = \ell }} h_{a} x_1^{a_1}x_2^{a_2}\ldots x_m^{a_m}  \in ker(\iota_\Gamma).\]
By definition, \(ker(\iota_\Gamma)\) vanishes at \(\Gamma(\C^s)\), implying \(h_\Gamma(\Gamma(\C^s))=\{0\}\). 
Let \(p > 2^{m}\) be a prime, fix a primitive \(p\)-th root of unity \(\zeta \in \C\) and a subset \(B\subset \{0,1,\ldots,p-1\}\) of size \(|B|=\dim(\mathcal{M}) = \binom{m}{\ell}\).
We claim that there is a \(b \in B\) such that \(h_{\Gamma}(f(\zeta^b))\ne 0\). Assume otherwise, implying 
 \begin{equation}\label{equation-elusive-curve}
    0=h_\Gamma\left(f\left(\zeta^{b}\right)\right)   = \sum_{\substack{a \in \{0,1\}^m
    \\  d_H(a)= \ell }} h_{a} \zeta^{b(a_1+a_22+a_32^2+\ldots+a_m2^{m-1})},\ \forall b\in B.
 \end{equation}
 By the uniqueness of binary expansion, the exponents of \(\zeta^b\) in the above sum are all distinct, and the set of exponents 
 \[E:=\left\{a_1+a_22+a_32^2+\ldots+a_m2^{m-1} \middle| a \in \{0,1\}^m, d_H(a)=\ell \right\}\] has size \(\binom{m}{\ell}\). Further, the maximum element in \(E\) is at most \(p-1\). Therefore, we can relabel the summation in equation \ref{equation-elusive-curve} as
 \[\sum_{t\in E} \widehat{h}_t\zeta^{bt} = 0,\ \forall b \in B\]
 where at least one of the complex coefficients in \((\widehat{h}_t)_{t\in E}\) is non zero. Hence, the square matrix \(\left(\zeta^{bt}\right)_{b \in B,t\in E}\) has determinant 
\[\det\left(\left(\zeta^{bt}\right)_{b \in B,t\in E}\right)=0,\]
contradicting Chebotarev's theorem on roots of unity \cite{sl}. 
Therefore our assumption is wrong and there is indeed a \(b \in B\) such that \(h_{\Gamma}(f(\zeta^b))\ne 0\), proving \(f\) is \((s,2)\)-elusive. 
    \qed
\end{proof}
\begin{remark}\label{remark-low-degree-annihilators}
    The curve defined by \(z\xmapsto{f} (z,z^2,z^4,\ldots,z^{2^{m-1}})\) satisfies degree two equations in the coordinate ring. For instance, for every \(i\in[m-1]\), the polynomial \(x_i^2-x_{i+1}\) vanishes at \(f(\C)\). Further \(f(\C)\) is contained in the image of  
    \[(z_1,z_2,\ldots,z_{m-1}) \xmapsto{\Gamma} (z_1,z_2,\ldots,z_i,z_i^2,z_{i+1},\ldots,z_{m-1}),\]
    implying \(f\) is not \((m-1,2)\)-elusive. This is compatible with \(f\) being \((\lfloor m^{9/10}\rfloor,2)\)-elusive, since \(\Gamma(\C^{m-1})\) is a hypersurface (that is, \((m-1)\)-dimensional), and not \(\lfloor m^{9/10}\rfloor\)-dimensional. The complete intersection defined by the polynomials \(\{x_i^2-x_{i+1},i\in[m-\lfloor m^{9/10}\rfloor]\}\) is an \(\lfloor m^{9/10}\rfloor\)-dimensional variety containing \(f(\C)\), but theorem \ref{theorem-elusive-curve-intro} implies that it does not have a degree-\(2\) parametric equation. More generally, we warn that the \((s,r)\)-elusive functions that we construct may have low degree (even lower than \(r\)) annihilators in the coordinate ring, but this is not in contradiction with being \((s,r)\)-elusive (for small enough \(s\ll m-1\)).
\end{remark}
\begin{remark}
Explicitness of elusive functions \(f:\C^n\longrightarrow\C^m\) is defined later (see definitions \ref{definition-explicit-elusive} and \ref{definition-polydefinable}), but only for degree-\(n^{O(1)}\) maps. The exponential degree elusive curve in theorem \ref{theorem-elusive-curve-intro} does not fit into this. Using multilinearisation, the exponential degree can be reduced at the cost of increasing dimensions \cite[Prop. 1.2, \S~1.5]{raz}. The multilinearisation of the curve in theorem \ref{theorem-elusive-curve-intro} is only poly(\(m\))-definable, insufficient for lower bounds.  An ``explicit'' \((m^{9/10},2)\)-elusive curve of degree \(2^{m^{o(1)}}\) has poly(\(m^{o(1)}\))-definable  multilinearisation, sufficient to conclude VP\(\ne\)VNP. 
\end{remark}
\section{Elusive monomial functions for low depths}\label{section-elusive}
In this section, we construct elusive functions with the number of coordinates \(m=n^2\) being a perfect square. The dimension of the elusive functions will be close to \(n\), but not quite \(n\). We use \(n'\) to denote the dimension and construct elusive functions from \(\C^{n'}\) to \(\C^m\). 
\subsection{Sumset expansion implies elusiveness}\label{subsection-elusive}
\begin{definition}\label{definition-expanding-exponent}
    Let \(m,n',\ell, C\) be positive integers with \(n'< m\). Call a matrix \[A=\left(\alpha_{i,j}\right)_{(i,j) \in [m] \times [n']} \in \Z_{\ge 0}^{m \times n'}\]  
    an \textbf{\((m,n',\ell, C)\)-expander} if there exists a vector \(\beta=(\beta_j)_{j \in [n']} \in \Z_{>0}^{n'}\)  
    such that \[A^\beta:=\left(\sum_{j=1}^{n'}\alpha_{i,j}\beta_j\right)_{ i \in [m]} \in \Z_{\ge0}^{m} \]
    expands as \(|\ell^{\le}A^\beta|\ge C\). By abuse of notation, we consider the vector \(A^\beta\) as a set while taking sumsets or indexing monomials. 
    We will refer to \(A\) as the \textbf{exponent matrix}, \(\beta\) as the \textbf{big vector}, and \(A^\beta\) as the \textbf{evaluation vector} with respect to \(\beta\). 
\end{definition}

\begin{lemma}\label{lemma-elusive}
Let \(n,d\) be positive integers with \(d>5\), set \(m:=n^2\) and \(s:=\lfloor n^{1+\frac{1}{d}}\rfloor\). 
For every \((m,n',n,n^n)\)-expander \(A=(\alpha_{i,j})\), 
    \begin{align*}
        f_{A}:\C^{n'} &\longrightarrow \C^m\\
        (z_1,z_2,\ldots,z_{n'}) & \longmapsto \left(\prod_{j=1}^{n'}z_j^{\alpha_{i,j}}\right)_{1\le i \le m}
    \end{align*}
   is \((s,d)\)-elusive.  
\end{lemma}
\begin{proof} 
Since \(A\) is an \((m,n',n,n^n)\)-expander, there exists a big vector \(\beta=(\beta_j)_{j \in [n']} \in \Z_{>0}^{n'}\) such that \(|n^\le A^\beta|\ge n^n\).
For every \(t \in n^\le A^\beta\), pick a sequence \((i_1(t),i_2(t),\ldots,i_{n_t}(t)) \in [m]^{n_t}\) of indices with \(n_t\le n\) such that 
\[\left(\sum_{j\in [n']}\alpha_{i_1(t),j}\beta_j\right)+\left(\sum_{j\in [n']}\alpha_{i_2(t),j}\beta_j\right)+\ldots+\left(\sum_{j\in [n']}\alpha_{i_{n_t}(t),j}\beta_j\right)=t.\] 
Here, \(n_t\) is the number of summands in the chosen sum to hit the target \(t\), which can be chosen to be at most \(n\) since \(t \in n^\le A^\beta\). %
Let \(\C[x_1,x_2,\ldots,x_m]\) denote the coordinate ring of \(\C^{m}\). 
For distinct \(t,t^\prime \in n^\le A^\beta\), the monomials \(x_{i_1(t)}x_{i_2(t)}\ldots x_{i_{n_t}(t)} \neq x_{i_1(t^\prime)}x_{i_2(t^\prime)}\ldots x_{i_{n_{t'}}(t^\prime)}\). Therefore the subset of monomials 
\[M:=\left\{x_{i_1(t)}x_{i_2(t)}\ldots x_{i_{n_t}(t)},t\in n^\le A^\beta\right\} \subset \C[x_1,x_2,\ldots,x_m]\] 
is a basis for the \(|n^\le A^\beta|\) dimensional \(\C\)-linear subspace \(\mathcal{M}\) of the coordinate ring it spans. For every degree-\(d\) map \(\Gamma:\C^s \longrightarrow \C^m\), we intend to construct a polynomial \(h_\Gamma\in \mathcal{M}\) that vanishes on \(\Gamma(\C^s)\), but not on \(f(\C^{n'})\), thereby proving the lemma.  
To this end, consider an arbitrary such degree-\(d\) map \(\Gamma:\C^s \longrightarrow \C^{m}\) taking 
\[       (z_1,z_2,\ldots,z_s)\longmapsto \left(g_{i}(z_1,z_2,\ldots,z_s)\right)_{i\in [m]}.
 \]
Let \(\iota_\Gamma:\mathcal{M}\longrightarrow\C[z_1,z_2,\ldots,z_s]\) denote the \(\C\)-linear map 
induced by the substitution \(x_i \longmapsto g_i(z_1,z_2,\ldots,z_s)\) for each of the \(i\in [m]\) coordinates. That is, for \(t \in n^\le A^\beta\), the \(t\)-th monomial is mapped as
\[x_{i_1(t)}x_{i_2(t)}\ldots x_{i_{n_t}(t)}  \longmapsto \prod_{j\in[n_t]}\left(g_{i_j(t)}(z_1,z_2,\ldots,z_s)\right).\]
Since \(\Gamma\) is a degree-\(d\) map and \(\mathcal{M}\) consists of degree at most \(n\) polynomials, the image \(\iota_\Gamma(\mathcal{M})\) is contained in the \(\binom{s+nd}{nd}\)-dimensional subspace of degree at most \(nd\) polynomials in \(\C[z_1,z_2,\ldots,z_s]\). 
Since \(A\) is an \((m,n',n,n^n)\)-expander,  
\begin{align*}\label{equation-elusive-proof-dimension-count}
    \dim(\mathcal{M}) = |n^\le A^\beta| \ge n^n > \left(\frac{2es}{nd}\right)^{nd} > \binom{s+nd}{nd}, 
\end{align*}
implying the kernel \(\ker(\iota_\Gamma)\) is at least one dimensional, assuring there 
exists a non zero  
\[h_\Gamma(x_1,x_2,\ldots,x_m) = \sum_{t \in n^\le A^\beta} h_t x_{i_1(t)}x_{i_2(t)}\ldots x_{i_{n_t}(t)}  \in ker(\iota_\Gamma).\]
By construction, \(h(\Gamma(\C^s))=\{0\}\). It remains to prove that \(h_\Gamma\) does not vanish on \(f_A(\C^{n'})\).
Let \(p\) be a prime bigger than the maximum in \(n^\le A^\beta\), to ensure \(|\{t \bmod p|t\in n^\le A^\beta\}|= |n^\le A^\beta|\).  
Fix a primitive \(p\)-th root of unity \(\zeta \in \C\) and a subset \(B\subset \{0,1,\ldots,p-1\}\) of size \(|B|=|n^\le A^\beta|\). 
We claim there is a \(b\in B\) such that \(h_\Gamma\left(f_A\left(\zeta^{b\beta_1},\zeta^{b\beta_2},\ldots,\zeta^{b\beta_{n'}}\right)\right) \ne 0\). If the claim is false, 
 \begin{align*}
0=&h_\Gamma\left(f_A\left(\zeta^{b\beta_1},\zeta^{b\beta_2},\ldots,\zeta^{b\beta_{n'}}\right)\right)   =  \sum_{t \in n^\le A^\beta} h_t \zeta^{\sum_{j=1}^{n'}\alpha_{i_1(t),j}b\beta_j}\zeta^{\sum_{j=1}^{n'}\alpha_{i_2(t),j}b\beta_j} \ldots \zeta^{\sum_{j=1}^{n'}\alpha_{i_{n_t}(t),j}b\beta_j} \\
     = & \sum_{t \in n^\le A^\beta} h_t \zeta^{b\left[\left(\sum_{j=1}^{n'}\alpha_{i_1(t),j}\beta_j\right)+\left(\sum_{j=1}^{n'}\alpha_{i_2(t),j}\beta_j\right) +\ldots+\left(\sum_{j=1}^{n'}\alpha_{i_{n_t}(t),j}\beta_j\right)\right]} = \sum_{t \in n^\le A^\beta} h_t \zeta^{bt},\ \ \forall b\in B.
 \end{align*}
Since at least one coefficient in \(\{h_t,t\in n^\le A^\beta\}\) is non zero, the square matrix \(\left(\zeta_p^{bt}\right)_{b \in B,t\in n^\le A^\beta}\) has determinant 
\[\det\left(\left(\zeta_p^{bt}\right)_{b \in B,t\in n^\le A^\beta}\right)=0.\]
But this contradicts Chebotarev's theorem on roots of unity \cite{sl}. Hence, the claim  that there is a \(b\in B\) such that \(h_\Gamma\left(f_A\left(\zeta^{b\beta_1},\zeta^{b\beta_2},\ldots,\zeta^{b\beta_{n'}}\right)\right) \ne 0\) is true, implying \(f_A\) is \((s,d)\)-elusive. 
\qed
\end{proof}
\begin{remark}\label{remark-positive-characteristic}
With lemma \ref{lemma-elusive} as a representative example, we next discuss how to extend our results to fields of large positive characteristic. Let \(K\) be a field of characteristic \(q\). Consider the order \(p\) Fourier matrix \((\zeta^{ij})_{i,j\in\{0,1,\ldots,p\}}\) in the proof of lemma \ref{lemma-elusive}, but over \(K\). That is, \(p\) is a prime bigger than the maximum element in \(n^\le A^\beta\) and \(\zeta\) is a primitive \(p\)-th root of unity in the algebraic closure of \(K\). We appeal to Zhang's finite field analogue of Chebotarev's theorem on roots of unity \cite[Thm. A]{zha} to prove the theorem over \(K\), but Zhang's theorem is restrictive in two ways. First, it only works for prime \(p\) order Fourier matrices when \(q\) is a primitive root modulo \(p\).  
The second, more severe restriction is that the characteristic \(q\) needs to be exponentially big compared to the order \(p\) of the matrix. See \cite[Eqn. 3.1]{zha} for precise bounds, but in our contexts, \(q =\Omega(p^p)\) is necessary. 
Hence, we can prove analogues of our results in positive characteristic, but only with the characteristic \(q\) double exponential in the input size. Further, if we want to claim lower bounds for a particular (family of) \(q\),  we must assume (the quantitative)  Artin's primitive roots conjecture, to ensure the existence of primes \(p\) such that \(n^n \le p \le \log_p(q)\) and \(q\) is a primitive root modulo \(p\). The existence of such pairs of primes \(p,q\) is not known, even assuming the generalized Riemann hypothesis (see \cite{fp} for bounds that come close).
A curious question is if our methods can work over smaller characteristic, but it likely needs new ideas.   
\end{remark}
\begin{definition}\label{definition-polydefinable}
    An \(f\in\C[z_1,z_2,\ldots,z_n]\) is \textbf{poly\((n)\)-definable} if there exists a size \(n^{O(1)}\) circuit computing a polynomial \(\widehat{f}\in \C[x_1,x_2,\ldots,x_n,e_1,e_2,\ldots,e_a]\) (for some \(a=n^{O(1)}\)) such that 
     \[f(z_1,z_2,\ldots,z_n) = \sum_{(e_1,e_2,\ldots,e_a)\in \{0,1\}^a} \widehat{f}\left(z_1,z_2,\ldots,z_n,e_1,e_2,\ldots,e_a\right).\]  
  In contexts where the number of variables is clear, by poly-definable we mean poly\((n)\)-definable. 
\end{definition}
\begin{definition}\label{definition-explicit-elusive}
    Following \cite[Def. 1.3]{raz},  an \textbf{explicit} presentation of a function \(f:\C^n\longrightarrow\C^m\) mapping \((z_1,z_2,\ldots,z_n)\longmapsto (f_i(z_1,z_2,\ldots,f_n))_{i\in[m]}\) is a size \((n\log m)^{O(1)}\)  arithmetic circuit for a polynomial \(\widehat{f} \in \C[z_1,z_2,\ldots,z_n,e_1,e_2,\ldots,e_a,w_1,w_2,\ldots,w_b]\) such that 
    \[f_i(z_1,z_2,\ldots,z_n) = \sum_{(e_1,e_2,\ldots,e_a)\in \{0,1\}^a} \widehat{f}\left(z_1,z_2,\ldots,z_n,e_1,e_2,\ldots,e_a,i_1,i_2,\ldots,i_b\right),\ \ \forall i\in[m] ,\]  
    where \(a=n^{O(1)}\), \(b=\lceil\log_2m\rceil\) and \((i_1,i_2,\ldots,i_b)\) is the binary expansion of \(i\). 
    In particular, for an explicit presentation of \(f\), it is not sufficient for each coordinate \(f_i\) of \(f\) to be poly\((n\log m)\)-definable, there must be a uniform circuit whose specialisations define each of the coordinates, as above.  
\end{definition}
\begin{proof} [of theorem \ref{theorem-reduction-intro}]
    Following the outline  in \S~\ref{subsection-intro-expansion-elusive}, it is easy to show that the monomial map associated with \(A\) is \((s,r)\)-elusive. We skip the details as it is virtually identical to the proof of lemma \ref{lemma-elusive}, changing parameters mutatis mutandis. Explicitness is shown as in the proof of theorem \ref{theorem-lower-bound-intro}, with one added detail. Since the exponents are bounded \(\alpha_{(i,j)}=2^{n^{O(1)}}\) exponentially in \(n\), using a multilinearisation trick, one can lower the degree by introducing new variables resulting in a poly(\(\Theta(n)\))-definable map \cite[Prop. 1.2 and \S~1.5]{raz}. Therefore the theorem follows from \cite[Cor. 3.9]{raz}. \qed
\end{proof}
\subsection{Explicit polynomials needing superlinear size sublogarithmic depth circuits}\label{subsection-lower-bounds}
In this subsection, the number of coordinates \(m=n^2\) will be a perfect square. For ease of  describing the explicit polynomials, we will index the \(n^2\) coordinates as an \(n \times n\) matrix. 
\begin{definition}\label{definition-explicit}
    Let \(n',m\) be positive integers with \(n'< m\). Define an exponent matrix 
    \(A=\left(\alpha_{i,j}\right)_{i \in [m], j \in [n']}\in \Z_{\ge0}^{m \times n'}\) to be
    \begin{itemize}
        \item \textbf{explicit} if there is a deterministic polynomial time (in \(n'\log m\)) Turing machine that takes as inputs \(i\in [m], (\alpha'_{j})_{j\in[n']} \in \Z_{>0}^{n'}\) and outputs a bit that equals one if and only if \(\alpha'_{j} = \alpha_{i,j}, \forall j \in [n']\).
        \item \textbf{strongly explicit} if there is a deterministic polynomial time (in \(\log m\)) Turing machine that takes as inputs \(i\in [m],j\in[n'], \alpha'_{i,j}\in \Z_{>0}^{m\times n'}\) and outputs a bit that equals one if and only if \(\alpha'_{i,j} = \alpha_{i,j}\).
    \end{itemize}
\end{definition}
\begin{lemma}\label{lemma-construction-main}
Let \(c\) be  positive integer and set \(n:=2^c\). 
The exponent matrix 
    \[A:= \left(\gamma_{v,k}\delta_{(u,w)}\right)_{(u,v) \in [n]\times [n],(w,k) \in [n]\times [c+1]}\]
 is a strongly explicit \((n^2,n(c+1),n,n^n)\)-expander, where \(\sum_{k=1}^{c+1}\gamma_{v,k}2^{k-1}=v\)
is the binary expansion of \(v\) and \(\delta_{(,)}\) is the Kronecker delta function. 
\end{lemma}
\begin{proof}
Consider the big vector \[\beta= (\beta_{w,k})_{(w,k) \in [n]\times [c+1]}:= (2^{k-1}n^{w-1})_{(w,k) \in [n]\times [c+1]}.\] 
The corresponding evaluation vector (indexed as a matrix) is \[A^{\beta} = \left(vn^{u-1}\right)_{(u,v) \in [n]\times [n]},\] 
since for every \((u,v)\in [n]\times[n]\), 
\begin{equation*}\label{equation-design-exponents}
    \sum_{(w,k)\in [n]\times [c+1]} \gamma_{v,k}\delta_{(u,w)}\beta_{w,k}= \sum_{k \in [c+1]}\gamma_{v,k}\beta_{u,k} =\sum_{k \in [c+1]}\gamma_{v,k}2^{k-1}n^{u-1}=vn^{u-1}.
\end{equation*}
It is illustrative to write the evaluation vector down as the matrix:
\begin{equation}\label{equation-design-vector}
   A^\beta = 
\begin{pmatrix}
1 & 2 & \ldots & n\\
n & 2n & \ldots & n^2\\
\vdots & \vdots & \ddots & \vdots\\
n^{n-1} & 2n^{n-1} & \ldots & n^{n}\\
\end{pmatrix}. 
\end{equation}
Every number of the form \(v n^{u-1}, (v,u) \in [n-1]\times [n]\) is an entry in \(A^\beta\). Hence, by \(n\)-ary expansion, every number in \([n^n-1]\) can be expressed as a sum of at most \(n\) entries from \(A^\beta\). Hence, \(|n^\le A^\beta|\ge n^n-1\). Also, \(n^n \in A^\beta\) implies \(n^n \in n^\le A^\beta\).  
Hence \(|n^\le A^\beta|\ge n^n\) and \(A\) is an \((n^2,n(c+1),n,n^n)\)-expander. Strong explicitness of \(A\) is clear, since given \((u,v),(w,k)\), to decide the predicate \(\gamma_{v,k}\delta_{(u,w)}\) only involves the binary expansion of \(v\le n\) and a couple of comparisons.
\qed
\end{proof}

\begin{proof} (of theorem \ref{theorem-lower-bound-intro})
Fix a \((c,n)\) as in the statement and consider the exponent matrix \(A\) from lemma \ref{lemma-construction-main}. For every \((w,k)\in[n]\times [c+1]\), introduce a variable \(x_{w,k}\). For every \((u,v) \in [n] \times [n]\), define the monomial 
\[f_{(u,v)}:= \prod_{(w,k)\in [n] \times [c+1]}x_{(w,k)}^{\alpha_{(u,v),(w,k)}} = \prod_{(w,k)\in [n] \times [c+1]}x_{(w,k)}^{\gamma_{v,k}\delta_{(u,w)}} \]
corresponding to the \((u,v)\)-th row of \(A\). Each \(f_{(u,v)}\) is multilinear, since the \(\gamma_{v,k}\delta_{(u,w)}\) are zero or one. Further, the degree of \(f_{(u,v)}\) is at most \(c+1\), despite appearing to be a product of \(n(c+1)\) terms, due to the sparsity induced by \(\delta_{(u,w)}\). 
By lemma \ref{lemma-construction-main}, \(A\) is an \((n^2,n(c+1),n,n^n)\)-expander, which is sufficient expansion for lemma \ref{lemma-elusive} to conclude that 
\begin{align*}
    f_A:\C^{n(c+1)} &\longrightarrow \C^{n\times n} \\
(x_{w,k})_{(w,k)\in [n]\times [c+1]} &\longmapsto \left(f_{(u,v)}\right)_{(u,v) \in [n] \times [n]}
\end{align*}
is \((\lfloor n^{1+1/d}\rfloor,d)\)-elusive. 
We now appeal to Raz's framework to convert the elusive function \(f_A\) to an explicit polynomial with circuit lower bound, as is done in \cite[\S~4.2]{raz}. First, we use \cite[Prop. 3.11]{raz} to package the elusive function into a set of explicit polynomials as follows. Introduce a new set of variables \(y_1,y_2,\ldots,y_n\) and in accordance with \cite[\S~3.3]{raz} define 
\[\widetilde{f}_{u}:= \sum_{v\in [n]} y_vf_{(u,v)}, \ \forall u \in [n].\]
By \cite[Prop. 3.11]{raz}, \(f_A\) being \((\lfloor n^{1+1/d}\rfloor,d)\)-elusive implies that every depth \(d\) arithmetic circuit over \(\C\) to compute the \(n\) polynomials \((\widetilde{f}_u)_{u\in[n]}\) has size at least \(n^{1+1/d}\). Second, to arrive at a single polynomial to which the circuit lower bound applies, we look to the Baur-Strassen theorem \cite{bs}(see also \cite{mor}), which states that if a polynomial can be computed by a size \(s'\) and depth \(d'\) arithmetic circuit, then there is an arithmetic circuit of size \(5d'\) and depth \(3d'\) that computes all its partial derivatives. Introduce \(n\) new variables \(z_1,z_2,\ldots,z_n\) and set 
\[  \widetilde{f}:= \sum_{u \in [n]} z_u \widetilde{f}_{u}\Rightarrow \frac{\partial\widetilde{f}}{\partial z_u} = \widetilde{f}_u ,\ \forall u \in [n].\]
Therefore, every depth \(\lfloor d/3\rfloor\) arithmetic circuit over the complex numbers that computes \(\widetilde{f}\) has  size at least \((n^{1+1/d})/5\). 
Expanding the definition, \(\widetilde{f}\) is indeed the polynomial 
\begin{equation*}
    \widetilde{f} = \sum_{u \in [n]} z_u \widetilde{f}_{u} = \sum_{u \in [n]} z_u \sum_{v\in [n]} y_v \prod_{(w,k)\in [n] \times [c+1]}x_{(w,k)}^{\gamma_{v,k}\delta_{(u,w)}}
\end{equation*}
in the theorem statement. All that remains is to prove that \(\widetilde{f}\) is poly-definable. By lemma \ref{lemma-construction-main}, \(A\) is explicit (in fact, strongly) in the sense of definition \ref{definition-explicit}. Since \(f_A\) is multilinear with only zero/one coefficients, \(A\) being explicit implies that \(f_A\) is poly-definable (see end of \cite[Sec. 1.5]{raz}). Since  \(f_A\) is poly-definable, so are \((\widetilde{f}_u)_u\), and so is \(\widetilde{f}\) (by \cite[Prop. 3.6]{raz}). 
\qed
\end{proof}

\subsection{Elusive functions from error correcting codes}\label{subsection-codes}
\begin{definition}\label{definition-codes}
    Let \(\F_2\) be the field with two elements. A (binary linear) code of length \(m\) is an \(\F_2\)-linear subspace of \(\F_2^m\) and its elements are called codewords. The minimum distance of a code \( C\subseteq \F_2^m\) (of non zero dimension) is \(\min_{ x\in C, x\ne 0^m}d_H(x),\) where we recall \(d_H\) is the Hamming weight. A flat matrix \(H=(h_{j,i}) \in \F_2^{n'\times m}\)  defines a code \(C_H:=\{x\in\F_2^m|Hx=0\}\). For a code \(C\), every matrix \(H\) such that \(C=C_H\) is called a parity check matrix of \(C\). 
    A (sparse) vector in \(x \in \F_2^m\) is called an error vector and \(Hx\in\F_2^{n'}\) is its syndrome.
\end{definition}
\begin{lemma}\label{lemma-elusive-code}
Let \(n',\ell,m\) be positive integers with \(n'<m\).  Let \((\alpha_{i,j})_{i\in[m],j\in[n']} \in \{0,1\}^{m\times n'}\) be an exponent matrix and \((\hat\alpha_{i,j})_{i,j} \in \F_2^{m\times n'}\) its reduction modulo two. 
    If the code
       \[\left\{\hat\gamma_{i} \in \F_2^{m} \middle| \sum_{i\in [m]} \hat\alpha_{i,j}\hat\gamma_{i} = 0,\ \forall j \in [n'] \right\}\]
     (with the transpose of \((\hat\alpha_{i,j})_{i,j}\) as a parity check matrix) has minimum distance at least \(2\ell+1\), then \((\alpha_{i,j})_{i,j}\) is an \((m,n',\ell,\binom{m+\ell}{\ell})\)-expander.
\end{lemma}
\begin{proof}
Set \(\beta_j:=(\ell+1)^{j-1}\) for \(j\in[n']\), so that the big vector \(\beta=(\beta_j)_{j \in [n]}\) is the basis of \((\ell+1)\)-ary expansion. Let 
    \(E:=\left\{\nu=(\nu_{i})_{i\in [m]}\in\{0,1\}^{m}\mid d_H(\nu)\le \ell\right\}\) and consider the map \(\pi:E \longrightarrow \Z\) taking 
    \[\nu \longmapsto \sum_{i\in[m]}\sum_{j\in[n']}\nu_{i}\alpha_{i,j}\beta_j.\]
    We claim for distinct \(\nu,\mu\in E\) that \(\pi(\nu)\neq\pi(\mu)\). To this end, consider \(\nu,\mu\in E\) such that 
    \begin{equation*}
\sum_{i\in[m]}\sum_{j\in[n']}\nu_{i}\alpha_{i,j}\beta_j=\sum_{i\in[m]}\sum_{j\in[n']}\mu_{i}\alpha_{i,j}\beta_j.
    \end{equation*}
    Since \(\sum_{i\in[m]}\nu_{i}\alpha_{i,j}\le \ell\) and \(\sum_{i\in[m]}\mu_{i}\alpha_{i,j}\le \ell\),
    by the uniqueness of \((\ell+1)\)-ary expansion, 
    \begin{equation}\label{equation-syndrome-integers}
        \sum_{i\in[m]}\nu_{i}\alpha_{i,j}=\sum_{i\in[m]}\mu_{i}\alpha_{i,j},\ \ \forall j\in[n'].
    \end{equation}
    Denoting \(\hat\nu_{i},\hat\mu_{i}\in \F_2\) to be the respective reductions of \(\nu_{i},\mu_{i}\) modulo \(2\), equation \ref{equation-syndrome-integers} implies that the reduced vectors \((\hat\nu_{i})_{i\in[m]}\) and \((\hat\mu_{i})_{i\in[m]}\) satisfy
        \begin{equation}\label{equation-syndrome}        \sum_{i\in[m]}\hat\nu_{i}\hat\alpha_{i,j}=\sum_{i\in[m]}\hat\mu_{i}\hat\alpha_{i,j},\ \ \forall j\in[n'].
    \end{equation}
    In coding terms, this is stating that the error vectors \((\hat\nu_{i})_{i\in[m]}\) and \((\hat\mu_{i})_{i\in[m]}\) have the same syndrome. Therefore, their difference 
    \((\hat\nu_{i})_{i\in[m]} - (\hat\mu_{i})_{i\in[m]}\) 
    is a codeword. Since, each of the error vectors is of weight at most \(\ell\) and the minimum weight of the code is at least \(2\ell+1\), the difference of the error vectors is the all zero codeword. Hence, 
    \( (\hat\nu_{i})_{i\in[m]} = (\hat\mu_{i})_{i\in[m]}\), which implies that their integer lifts are also equal. That is, \(\nu = \mu\), proving the claim. Since the weight of each \(\nu\in E\) is  at most \(\ell\), \(\{\pi(\nu)|\nu \in E\} \subseteq \ell^\le A^\beta\). Since the number of weight at most \(\ell\) error vectors is \( \binom{m+\ell}{\ell}\), we have \(|\ell^\le A^\beta|\ge \binom{m+\ell}{\ell}\) and the lemma follows. 
\qed
\end{proof}
\begin{remark}\label{remark-codes}
In light of lemma \ref{lemma-elusive-code}, explicit binary linear codes of length \(m\), dimension \(n'\) and minimum distance \(2\ell+1\) yield explicit  \((m,n',\ell,\binom{m+\ell}{\ell})\)-expanders.  
To control the degree of the resulting elusive functions (and the explicit polynomials), we look to parity check matrices of expander codes built from unbalanced bipartite expander graphs \cite{ss-expander}. Consider a right \(c\)-regular bipartite graph with \(n'\) parity constraints on the left and \(m\) coordinates/variables on the right. Take its \(n'\times m\) left-right adjacency matrix  as the parity check matrix.
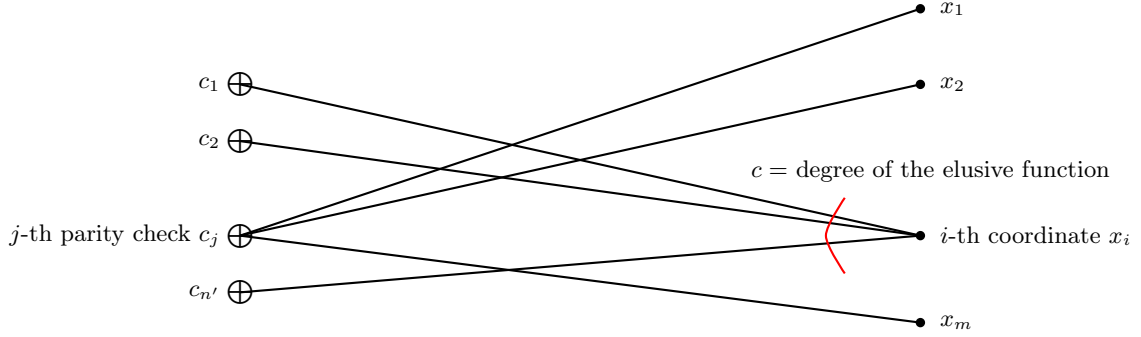
\begin{figure}[H]
    \centering
    \begin{tikzpicture}[thick,scale=1, every node/.style={transform shape}]
    \tikzstyle{vertex}=[circle, fill, inner sep=1.25pt]
    \node [vertex] at (9,.85) {};
    \node [right=4pt] at (9,.85) {\(x_{m}\)};
    \node [vertex] at (9,2) {};
    \node [right=4pt] at (9,2) {\(i\)-th coordinate \(x_i\)};
    \node [vertex] at (9,4) {};
    \node [right=4pt] at (9,4) {\(x_2\)};
    \node [vertex] at (9,5) {};
    \node [right=4pt] at (9,5) {\(x_1\)};
    \node at (0,4) {\(\bigoplus\)};
    \node [left=4pt] at (0,4) {\(c_1\)};
    \node  at (0,3.25) {\(\bigoplus\)};
    \node [left=4pt] at (0,3.25) {\(c_2\)};
    \node  at (0,2) {\(\bigoplus\)};
    \node [left=4pt] at (0,2) {\(j\)-th parity check \(c_j\)};
    \node at (0,1.25) {\(\bigoplus\)};
    \node [left=4pt] at (0,1.25) {\(c_{n'}\)};
%
%
%
    \draw (9,2) -- (0,3.25);
    \draw (9,2) -- (0,4);
    \draw (9,2) -- (0,1.25);
    \draw (9,5) -- (0,2);
    \draw (9,4) -- (0,2);
    \draw (9,.85) -- (0,2);
    \draw [red] plot [smooth] coordinates {(8,1.5) (7.75,2) (8,2.5)};
   \node [right=4pt] at (6.5,2.85) {\(c=\) degree of the elusive function};

    \end{tikzpicture}
    \caption{A right \(c\)-regular bipartite graph whose adjacency defines the parity checks.}
    \label{fig:edge-vertex-incidence}
\end{figure} 
If every subset of the right vertices of size at most \(2\ell+1\) has at least \((1-o(1))c(2\ell+1)\) neighbours on the left, then the code that the parity check matrix defines has minimum distance at least \(2\ell+1\) \cite[Thm. 7]{ss-expander}. Such bipartite graphs are called lossless unbalanced expanders. Lossless refers to the expansion factor (1-o(1)) being close to one and unbalanced refers to the number of right vertices \(m\) being much larger than the number of left vertices \(n\).
There are explicit constructions of lossless expanders due to Capalbo, Reingold, Vadhan, and  Wigderson \cite{crvw}, or Guruswami, Umans, and Vadhan \cite{guv}, with the right degree \(c\) polylogarithmic in the imbalance \(m/n'\). The degree of the resulting 
elusive function is \(c\), and the explicit polynomials are of degree two more than \(c\)
 (see proof of theorem \ref{theorem-lower-bound-intro}). For \(m=n^2\) and \(n'=O(n\log n)\), we get explicit polynomials of degree polylogarithmic in \(n\).  We do not write out the details as the adhoc construction in lemma \ref{lemma-construction-main} resulting in theorem \ref{theorem-lower-bound-intro} gives the same superlinearity, but with a smaller (still logarithmic in \(n\)) degree. This connection with codes should be better investigated, to see if it applies to wider parameter ranges. The downside is that the minimum distance demanded in lemma \ref{lemma-elusive-code} is too large, since it is designed to have the maximum possible sumset expansion, without any collisions of error vectors to the same syndrome. List decoding does not seem to help at this high rate regime. A more relaxed/averaged analysis allowing for some syndrome collisions will be interesting. The promising aspect of the error correction code method is that in parameter ranges where it is easy to construct explicit codes, one can construct codes with a parity check matrix of low column weights, to get polylogarithmic degree explicit polynomials. 
\end{remark}
\section{Towards super-linear bounds via matrix rigidity}\label{section-rigidity}
\begin{proof}[\textit{of theorem \ref{theorem-rigidity-intro}}] 
Consider an \((n^2,n',n^3,n^{8s})\)-expander
\[A=(a_{(u,v),j})_{(u,v)\in[n]\times[n],j\in [n']}\]
for a large enough \(n\).
Let \(X=(x_j)_{j\in [n']}\) and \(Y=(y_v)_{v\in[n]}\) be two tuples of variables. 
Associate to \(A\), the symbolic matrix 
\[f_A(X):= \left( \prod_{j\in[n']}x_j^{\alpha_{(u,v),j}}\right)_{(u,v)\in[n]\times[n]} \in \C[x_1,x_2,\ldots,x_{n'}]^{n\times n}\]
and denote by 
\[f_A(X)\times Y := \left(\sum_{v\in[n]} y_v\prod_{j\in[n']}x_j^{\alpha_{(u,v),j}} \right)_{u\in[n]} \in \C[x_1,x_2,\ldots,x_{n'},y_1,y_2,\ldots,y_n]^{n} \]
the symbolic matrix-vector multiplication. 
Since \(A\) is an \((n^2,n',n^3,n^{8s})\)-expander, there exists a big vector \(\beta=(\beta_j)_{j \in [n']} \in \Z_{>0}^{n'}\) such that 
\[A^\beta = \left(\sum_{j\in [n']}\alpha_{(u,v),j}\beta_j\right)_{(u,v) \in [n]\times [n]}\]
expands as  \(|{n^3}^\le A^\beta|\ge n^{s}\).  
For every \(t \in {n^3}^\le A^\beta\), pick an \(n_t\)-tuple  \((i_1(t),i_2(t),\ldots,i_{n_t}(t)) \in [n\times n]^{n_t}\) of matrix indices with \(n_t\le n^3\) such that 
\[\left(\sum_{j\in [n']}\alpha_{i_1(t),j}\beta_j\right)+\left(\sum_{j\in [n']}\alpha_{i_2(t),j}\beta_j\right)+\ldots+\left(\sum_{j\in [n']}\alpha_{i_{n_t}(t),j}\beta_j\right)=t.\] 
Here, \(n_t\) is the number of summands in the chosen sum to hit the target \(t\), which can be made at most \(n^3\) since \(t \in\ {n^3}^\le A^\beta\). %
Let \(\C[w_{(1,1)},w_{(1,2)},\ldots,w_{(n,n)}]\) denote the coordinate ring of \(\C^{n\times n}\).  
For distinct \(t,t^\prime \in\ {n^3}^\le A^\beta\), the monomials \(w_{i_1(t)}w_{i_2(t)}\ldots w_{i_{n_t}(t)} \neq w_{i_1(t^\prime)}w_{i_2(t^\prime)}\ldots w_{i_{n_{t'}}(t^\prime)}\). Therefore the (not necessarily multilinear) monomials 
\[M:=\left\{w_{i_1(t)}w_{i_2(t)}\ldots w_{i_{n_t}(t)},t\in {n^3}^\le A^\beta\right\} \subset \C[w_{(1,1)},w_{(1,2)},\ldots,w_{(n,n)}]\] 
form a basis for the \(|{n^3}^\le A^\beta|\) dimensional \(\C\)-linear subspace \(\mathcal{M}\) of \(\C[w_{(1,1)},w_{(1,2)},\ldots,w_{(n,n)}]\) it spans. 
Consider the  degree at most \(s^4(s+1)\) universal map 
\begin{align*}
    U:\C^{2s}&\longrightarrow\C^{n\times n}\\
    (q_1,q_2,\ldots,q_{2s}) &\longmapsto \left(g_{(u,v)}(q_1,q_2,\ldots,q_{2s})\right)_{(u,v))\in[n] \times [n]}
\end{align*}
(defined in \cite[Lem. 4.1]{vk}) of Volk and Kumar whose image contains all the matrices whose linear transformation can be computed by some linear circuit of size at most \(s\).  
Let \(\iota_U:\mathcal{M}\longrightarrow\C[q_1,q_2,\ldots,q_{2s}]\) be the \(\C\)-linear map 
induced by the substitution \(w_{(u,v)} \longmapsto g_{(u,v)}(q_1,q_2,\ldots,q_{2s})\) for each of the \((u,v)\in [n]\times[n]\) coordinates. That is, for \(t \in {n^3}^\le A^\beta\), the \(t\)-th monomial is mapped as
\[w_{i_1(t)}w_{i_2(t)}\ldots w_{i_{n_t}(t)}  \longmapsto \prod_{j\in[n_t]}\left(g_{i_j(t)}(q_1,q_2,\ldots,q_{2s})\right).\]
Since \(U\) is a degree \(s^4(s+1)\) map and \(\mathcal{M}\) consists of degree at most \(n^3\) polynomials, the image \(\iota_U(\mathcal{M})\) is contained in the \(\binom{2s+s^4(s+1)n^3}{s^4(s+1)n^3}\)-dimensional subspace of degree at most \(s^4(s+1)\) polynomials in \(\C[q_1,q_2,\ldots,q_{2s}]\). 
Since \(A\) is an \((n^2,n',n^3,n^{8s})\)-expander,  
\begin{align*}
    \dim(\mathcal{M}) = | {n^3}^\le A^\beta| \ge n^{8s}>\binom{2s+s^4(s+1)n^3}{2s}= \binom{2s+s^4(s+1)n^3}{s^4(s+1)n^3},
\end{align*}
implying the kernel \(\ker(\iota_\Gamma)\) has dimension at least one, assuring there 
exists a non zero  
\[h_U(w_{(1,1)},w_{(1,2)},\ldots,w_{(n,n)}) = \sum_{t \in n^\le A^\beta} h_t w_{i_1(t)}w_{i_2(t)}\ldots w_{i_{n_t}(t)}  \in ker(\iota_U).\]
As in the proof of theorem \ref{theorem-lower-bound-intro}, for a prime \(p\) greater than the maximum element in \({n^3}^\le A^\beta\), we can show using Chebotarev's theorem on roots of unity that there exists a \(b\in \{0,1,\ldots,p-1\}\) such that 
\[h_U\left(f_A\left(\zeta^{b\beta_1},\zeta^{b\beta_2},\ldots,\zeta^{b\beta_{n'}}\right)\right)\neq 0.\]
Since \(h_U\in\ker(\iota_U)\), there is no size \(s\) linear circuit to compute the linear transformation
 \begin{equation}\label{equation-substitution-polynomial}
     f_A(X)/_{X\longmapsto \left(\zeta^{b\beta_1},\zeta^{b\beta_2},\ldots,\zeta^{b\beta_{n'}}\right)}\times Y \in \C[y_1,y_2,\ldots,y_n]^{n}.
 \end{equation}
Strassen proved  for every circuit computing a set of linear polynomials, there is a linear circuit computing the same polynomials, of size and depth a constant times that of the original circuit \cite{str-73}. In fact the underlying constant can be taken to be one \cite[Prop. 2.5]{raz}. Hence, there is no  circuit (linear or otherwise) of size \(s\) to compute the linear transformation in equation \ref{equation-substitution-polynomial}. Therefore, there is no circuit of size \(s\) to compute the \(n\)-tuple of polynomials  \(f_A(X)\times Y\). Otherwise, specialising  part of the input to the circuit under the substitution \(X\longmapsto \left(\zeta^{b\beta_1},\zeta^{b\beta_2},\ldots,\zeta^{b\beta_{n'}}\right)\) gives a size  \(s\) circuit for the linear transformation in equation \ref{equation-substitution-polynomial}.  
Introduce one last tuple \(Z=(z_u)_{u\in[n]}\) of variables and conclude that there is no circuit of size \(s/5\) to compute the polynomial 
\[\sum_{u\in[n]} z_u \left(f_A(X)\times Y\right)_u =\sum_{u\in[n]} z_u \sum_{v\in[n]} y_v\prod_{j\in[n']}x_j^{\alpha_{(u,v),j}}  \in \C[x_1,\ldots,x_{n'},y_1,\ldots,y_n,z_1,\ldots,z_n],\]
or else, Baur-Strassen computes all its derivatives \(f_A(Y)\times W\) at once in size \(s\) \cite{bs}.  
\qed
\end{proof}
\section{High-rank tensors with roots of unity coordinates}\label{section-tensors}
\subsection{Tensor rank, border rank and annihilators of secant varieties}\label{subsection-tensor-rank}
We refer to \cite{lan-book} for the tensor algebra basics used in this section. For positive integers \(n,d\), consider a tensor \(A\in \underbrace{\p^{n-1}\otimes \p^{n-1} \otimes \ldots \otimes\p^{n-1}}_{d}\) in the tensor product of \((n-1)\)-dimensional complex projective spaces. Fixing a coordinate system \(x^{(i)}=(x^{(i)}_j)_{j\in [n]}\) for the \(i\)-th copy of \(\p^{n-1}\), we can write \(A\) as \[\sum_{(j_1,j_2,\ldots,j_d)\in[n]^d}a_{j_1,j_2,\ldots,j_d}x^{(1)}_{j_1}\otimes x^{(2)}_{j_2}\otimes \ldots\otimes x^{(d)}_{j_d}.\]
We refer to \(a_{j_1,j_2,\ldots,j_d}\)s as the tensor coordinates, which are well defined up to a non zero constant. Therefore, abusing notation, we will also think of \(A\) as given by its coordinates \[A=(a_{j_1,j_2,\ldots,j_d})_{(j_1,j_2,\ldots,j_d)\in[n]^d} \in \p(\C^n\otimes \C^n \otimes \ldots \otimes \C^n).\] 
Rank one tensors are those of form \(u^{(1)}\otimes u^{(2)} \otimes \ldots \otimes u^{(d)}\) for some \((u^{(1)},u^{(2)}, \ldots , u^{(d)}) \in \p^{n-1}\times \p^{n-1} \times \ldots \times \p^{n-1} \). The tensor rank of \(A\) is the smallest number \(r\) such that \(A\) can be written as the sum of \(r\) rank one tensors.   
To define border rank, we look to the coordinate view, where rank one tensors are precisely the ones in the image of the Segre embedding (outer product)
\begin{align*}
    \p^{n-1}\times \p^{n-1} \times \ldots \times \p^{n-1} &\longrightarrow \p(\C^n\otimes \C^n \otimes \ldots \otimes \C^n)\\
    \left(u^{(1)},u^{(2)},\ldots,u^{(d)}\right)&\longmapsto \left(u^{(1)}_{j_1}u^{(2)}_{j_2}\ldots u^{(d)}_{j_d}\right)_{(j_1,j_2,\ldots,j_d)\in[n]^d}.
\end{align*}
This image (consisting of rank one tensors) is called the Segre variety \(X\). 
The \(r\)-th secant variety  
\[\sigma_r(X):= \overline{\bigcup_{(A_1,A_2,\ldots,A_r)\in X^r}\left(\C A_1+\C A_2+\ldots+\C A_r \right)},\]
 is the Zariski closure (denoted by the bar) of tensors of tensor rank at most \(r\). The border rank of a tensor \(A\) is defined as the smallest \(r\) such that \(A \in \sigma_r(X)\). For other equivalent definitions/characterizations of border rank, see \cite[15.4, 20.3]{bcs}. Clearly, the tensor rank of \(A\) is at least the border rank of \(A\). Therefore, to find tensors of high tensor rank, it suffices to find tensors of high border rank. As a projective subvariety of \(\p(\C^n\otimes \C^n \otimes \ldots \otimes \C^n)\), \(\sigma_r(X)\) satisfies homogeneous equations in the tensor coordinates \cite{lm}. If we can find a point outside \(\sigma_r(X)\), then that point corresponds to a tensor of border rank greater than \(r\). We will closely follow the technique of Volk and Kumar \cite{vk} to find such polynomials of small (polynomial in \(n^d\)) degree, and use our method to elude them. 
  Following \cite[\S~5]{vk}, we first write down a universal polynomial map encoding tensor rank \(r\) tensors. We will introduce \(ndr\) variables. The first \(nd\) of these are \(\left(u(1)^{(i)}_{j}\right)_{i\in[d],j\in[n]}\), thought of as encoding a symbolic rank one vector 
 \[\left(u(1)^{(1)}_j\right)_{j\in[n]} \otimes \left(u(1)^{(2)}_j\right)_{j\in[n]} \otimes \ldots \otimes \left(u(1)^{(d)}_j\right)_{j\in[n]}.\]
Likewise, for \(k\in[r]\), the \(k\)-th \(nd\)-tuple of variables \(\left(u(k)^{(i)}_{j}\right)_{i\in[d],j\in[n]}\) encodes the \(k\)-th rank one tensor. By construction, the degree-\(d\) map 
\begin{align}\label{equation-lowranktensormap}
     \phi: \C^{ndr} &\longrightarrow \p^{n-1}\otimes \p^{n-1} \otimes \ldots \otimes\p^{n-1} \\ 
    \left(u(k)^{(i)}_{j}\right)_{i\in[d],j\in[n],k\in[r]}&\longmapsto \sum_{k\in[r]}\left(\left(u(k)^{(1)}_j\right)_{j\in[n]} \otimes \left(u(k)^{(2)}_j\right)_{j\in[n]} \otimes \ldots \otimes \left(u(k)^{(d)}_j\right)_{j\in[n]}\right) \notag
\end{align}
 contains in its image \(\phi(\C^{ndr})\) (viewed as a set of tensor coordinates) all tensor of tensor rank at most \(r\). Let \(h\) be a homogeneous polynomial in the coordinate ring (of the tensor coordinates) that vanishes at \(\phi(\C^{ndr})\). Hence, \(h\) vanishes at the tensors of tensor rank at most \(r\). Since the set of zeroes of \(h\) is Zariski closed, by the minimality of Zariski closure, \(h\) must vanish at \(\sigma_r(X)\) (the Zariski closure of tensors of tensor rank at most \(r\)). Hence, if \(h\) does not vanish at a tensor, then that tensor has border rank greater than \(r\).\\

 Let us pause to estimate the maximum border rank attainable. The secant varieties form a proper chain of varieties \(\sigma_1(X)\subset\sigma_2(X)\subset\ldots\subset \sigma_s(X)\), until \(\sigma_s(X)\) fills the whole space of tensors. Call such an \(s\), which is the border rank such that \(\sigma_s(X)=\p(\C^n\otimes \C^n \otimes \ldots \otimes \C^n)\) but \(\sigma_{s-1}(X) \subset \p(\C^n\otimes \C^n \otimes \ldots \otimes \C^n)\), as the saturating border rank. Catalisano, Geramita and Gimigliano \cite[Rem 1.4]{cgg} showed that the saturating border rank is at least \(s \ge n^d/((n-1)d+1\), while Abo, Ottaviani and Peterson \cite[Thm. 5.2]{aop} showed that it is at most (for \(d>2\))
 \[s \ge \left\lfloor\frac{n^d}{(n-1)d+1}\right\rfloor + n - \left(\left\lfloor\frac{n^d}{(n-1)d+1}\right\rfloor \bmod n\right),\]
 where the last term is meant to be the remainder in \(\{0,1,\ldots,n-1\}\) of \(\lfloor n^d/((n-1)d+1) \rfloor\) divided by \(n\). We will choose \(r=\lfloor n^{d-1}/2d\rfloor \) as the target border rank and find a point outside \(\sigma_r(X)\), to get a tensor of border rank at worst a factor of half off the maximum possible. Further, for \(d=\Omega(\log n/\log\log n)\), the border rank  \(n^{d-1}/2d=n^{d(1-o(1))}\) is high enough for Raz's arithmetic formula lower bounds to apply.
\subsection{Semi-explicit high-rank tensors}\label{subsection-semiexplicit}
\begin{proof}[of theorem \ref{theorem-highrank-tensor-intro}]
    For ease of notation, flatten the \(n^d\) coordinates of \(\p(\C^n\otimes \C^n \otimes \ldots \otimes \C^n)\), so that tensor coordinates are indexed by \([n^d]\), say by the lexicographic ordering \([n]^d\longrightarrow[n^d]\). Let \(f:\C\longrightarrow\C^{n^d}\) denote the map \[z\longmapsto (z^{n^{2di}})_{i\in[n^d]}\] defining the coordinates of the symbolic tensor \((z^{n^{2di}})_{i\in[n^d]} \in \C[z]^{n^d}\).
    Let \(\C\left[x_i,i \in [n^d]\right]\) denote the homogeneous coordinate ring of the flattening of \(\p(\C^n\otimes \C^n \otimes \ldots \otimes \C^n)\). In a departure from the previous proofs, we work with its subring \(\Q\left[x_i,i \in [n^d]\right]\). Let \(\mathcal{M}\) denote the \(\Q\)-linear subspace of \(\Q\left[x_i,i \in [n^d]\right]\) generated by the basis  
    \[M:=\left\{x_1^{a_1}x_2^{a_2}\ldots x_{n^d}^{a_{n^d}} \middle| \ a=(a_1,a_2,\ldots,a_{n^d})\in \Z_{\ge 0}^{n^d}, \sum_{i\in[n^d]}a_i = n^{2d}\right\}\]
    of monomials of degree \(n^{2d}\), indexed by number strings \(a\) that sum to \(n^{2d}\). By the uniqueness of \(n^{2d}\)-ary expansion, distinct \(a,a'\) correspond to distinct   
    \(x_1^{a_1}x_2^{a_2}\ldots x_{n^d}^{a_{n^d}}\), \(x_1^{a'_1}x_2^{a'_2}\ldots x_{n^d}^{a'_{n^d}}\), implying \(\dim(\mathcal{M})=  \dbinom{n^{2d}+n^d-1}{n^{2d}}=  \dbinom{n^{2d}+n^d-1}{n^d-1}\).  
    Set 
    \(r=:\lfloor n^{d-1}/2d \rfloor \) and 
    recall the degree \(d\) universal map from equation \ref{equation-lowranktensormap} 
    \begin{align*}
        \phi:\C^{ndr} &\longrightarrow \p(\C^{n^d})\\
        (z_1,z_2,\ldots,z_{ndr}) &\longmapsto (\phi_i(z_1,z_2,\ldots,z_{ndr}))_{i\in[n^d]},
    \end{align*}
    whose image \(\phi(\C^{ndr})\) contains all rank at most \(ndr\) tensors. For \(i\in[n^d]\), we note that \( \phi_i\in \Q\left[x_i,i \in [n^d]\right]\), as its coefficients are all zero or one. 
    Let \(\iota:\mathcal{M}\longrightarrow\Q[z_1,z_2,\ldots,z_{ndr}]\) be the \(\Q\)-linear map 
induced by the substitution \(x_i \longmapsto \phi_i(z_1,z_2,\ldots,z_{ndr})\) for each of the \(i\in [n^d]\) coordinates. That is, 
the \(a\)-th monomial is mapped to
\[x_1^{a_1}x_2^{a_2}\ldots x_{n^d}^{a_{n^d}}  \longmapsto \prod_{i\in[n^d]}\left(u_i(z_1,z_2,\ldots,z_{ndr})\right)^{a_i}.\]
Since \(\phi\) is a degree-\(d\) map and \(\mathcal{M}\) consists of degree at most \(n^{2d}\) polynomials, the image \(\iota(\mathcal{M})\) is contained in the subspace of degree at most \(dn^{2d}\) polynomials in \(\Q[z_1,z_2,\ldots,z_{ndr}]\). The latter is a subspace of dimension 
\[ \dbinom{ndr+dn^{2d}}{ndr} < \left(\frac{en^{d+1}}{2}\right)^{\frac{n^d}{2}} \ll n^{d(n^{d}-1)} \le \dbinom{n^{2d}+n^d-1}{n^d-1} = \dim(\mathcal{M}),\]
where the \(``\ll''\) holds for large enough \(n^d\).  
Hence, there is a non zero homogeneous \[h= \sum_{a\mid \sum_{i\in[n^d]}a_i = n^{2d}} h_ax_1^{a_1}x_2^{a_2}\ldots x_{n^d}^{a_{n^d}} \in \ker(\iota),\] vanishing at the secant variety \(\sigma_{\lfloor n^{d-1}/2d\rfloor}(X)\) consisting of all tensors of border rank at most \(n^{d-1}/2d\). Let \(p> n^{2d(n^d+1)}\) be a prime and fix a primitive \(p\)-th root of unity \(\zeta\in \C\).  
As before (similar to the proof of theorem \ref{theorem-elusive-curve-intro}), there exists a \(b\in \{0,1,\ldots,p-1\}\) such that \(h(f(\zeta^b))\ne 0\), for otherwise Chebotarev's theorem on roots of unity is contradicted. 
Consider the field homomorphism \(\sigma:\Q(\zeta)\longrightarrow\Q(\zeta^b)\) taking \(\zeta \longmapsto \zeta^b\) and fixing \(\Q\). Since \(h \in \Q\left[x_i,i \in [n^d]\right]\), \[\sigma(h(f(\zeta)))=h(\sigma(f(\zeta))) = h(f(\zeta^b))\neq 0 \Rightarrow h(f(\zeta))\ne 0.\]
Hence the tensor with coordinates \((\zeta^{n^{2di}})_{i\in[n^d]}\) has border rank at least \((n^{d-1})/2d\). The proof holds even if the tensor coordinates are permuted, proving the theorem.
\qed
\end{proof}
\begin{remark}\label{remark-highrank-tensor-explicit}
    Finding explicit tensors of high rank is a notoriously difficult problem, even restricted to three dimensions. Part of the difficulty is that tensor ranks are NP-hard to compute \cite{has}. Further, the polynomials whose vanishing defines the secant variety (consisting of tensors of a bounded border rank) are intricate and difficult to study \cite{lm}, preventing techniques that try to find a point outside them.  In three dimensions, for \(n\times n \times n\) cubical formats, the typical tensor rank (or border rank) is \(\lfloor n^3/(3n-2) \rfloor = \Theta(n^2)\) \cite{lic}, but no explicit tensors of super-linear tensor rank \(\omega(n)\) are known. See \cite{amt,lan,ml} and the references therein for the best known constructions. In \(d\)-dimensional \(n\times n \times \ldots \times n\) formats, it is easy to construct tensors of rank  \(n^{\lfloor d/2\rfloor}\), by embedding a non singular \(n^{\lfloor d/2\rfloor} \times n^{\lfloor d/2\rfloor}\) matrix into a tensor \cite{raz-ten}. An explicit construction nearly doubling this bound is in \cite{amt}, beyond which no constructions are known. These fall well short of the \(n^{d(1-o(1))}\) rank demanded by arithmetic formula lower bounds. 
\end{remark}
Consider the semi-explicit construction of high rank tensors in \(d\)-dimensional \(n\times n \times \ldots \times n\) formats. To compare the construction in theorem \ref{theorem-highrank-tensor-intro} to the state of the art, we first observe that results of Volk-Kumar \cite{vk} and Kumar-Lokam-Patankar-Sharma \cite{klps} combine to give semi-explicit high border rank tensors with roots of unity coordinates.  
\begin{theorem}[Volk-Kumar+Kumar-Lokam-Patankar-Sharma]\label{theorem-tensor-rank-other}
    Let \(n,d\) be positive integers. Let \((p_{j_1,j_2,\ldots,j_d})_{(j_1,j_2,\ldots,j_d)\in[n]^d}\) be a tensor of distinct prime coordinates, each prime greater than \(n^{2d}\). Let \(\zeta_{p_{j_1,j_2,\ldots,j_d}}\in \C\) be a primitive \(p_{j_1,j_2,\ldots,j_d}\)-th root of unity. For large enough \(n^d\), the border rank of \(\left(\zeta_{p_{j_1,j_2,\ldots,j_d}}\right)_{(j_1,j_2,\ldots,j_d)\in[n]^d}\) is at least  \(n^{d-1}/2d\).
\end{theorem}
\begin{proof}
Consider \(d\)-dimensional \(n\times n \times \ldots \times n\) format tensors. 
    By \cite[Thm 1.4]{vk}, (with minor adjustments in parameters) there is a non-zero  polynomial \(h\) in the coordinate ring of the tensor of degree at most \(n^{2d}\), vanishing at border rank at most \(n^{d-1}/2d\) tensors. Further, (as evident from equation \ref{equation-lowranktensormap} and proof of theorem \ref{theorem-highrank-tensor-intro}) there exists such a polynomial \(h\) with rational coefficients. The theorem follows by applying \cite[Lem. 11]{klps} with the roots of unity \(\left(\zeta_{p_{j_1,j_2,\ldots,j_d}}\right)_{(j_1,j_2,\ldots,j_d)\in[n]^d}\) and the polynomial \(h\). 
    \qed
\end{proof}

\begin{remark}\label{remark-highrank-tensor}
An easy way to construct a semi-explicit high-rank tensor is to have algebraically independent numbers as coordinates: by definition they cannot satisfy a non zero polynomial (that vanishes on the secant variety of low border rank) over \(\Q\). We will not consider these as viable  avenues towards explicit constructions. 
Another way is to have entries that are sufficiently exponentially increasing integers. We will want constructions where this growth is as small as possible. 
Therefore we will restrict our attention to coordinates that are in an algebraic number field and small (in some notion of height), in the spirit of Strassen \cite{str-rational}. 
We will gauge the semi-explicitness of a tensor \((a_{j_1,j_2,\ldots,j_d})_{(j_1,j_2,\ldots,j_d)\in[n]^d} \in \p(\C^n\otimes \C^n \otimes \ldots \otimes \C^n)\) by the degree \([\Q\left((a_{j_1,j_2,\ldots,j_d})_{(j_1,j_2,\ldots,j_d)\in[n]^d)}\right):\Q]\) and the maximum (absolute Weil) height of the coordinates, wanting both to be small. 
Both theorems \ref{theorem-tensor-rank-other} and \ref{theorem-highrank-tensor-intro} give tensor coordinates of height one (since they are roots of unity), which is the best possible. 
Theorem \ref{theorem-tensor-rank-other} gives coordinates that are roots of unity in a cyclotomic field of degree  
\[\left[\Q\left(\left(\zeta_{p_{j_1,j_2,\ldots,j_d}}\right)_{(j_1,j_2,\ldots,j_d)\in[n]^d}\right):\Q\right] = \prod_{(j_1,j_2,\ldots,j_d)\in[n]^d}p_{j_1,j_2,\ldots,j_d}=\Omega\left( n^{2dn^d}(2d\log n)^{n^d}\right),\]
where the bound on the right follows from the prime number theorem. 
Theorem \ref{theorem-highrank-tensor-intro} gives tensor coordinates in a cyclotomic extension of degree \([\Q(\zeta_p):\Q]\) for a prime \(p>n^{2d(n^d+1)}\), which is \(O\left(n^{2d(n^d+1)}2dn^d\log n\right)\) by the prime number theorem for choosing the smallest such prime. Hence, theorem \ref{theorem-highrank-tensor-intro} is over a smaller degree field, by roughly a \(e^{n^d \log n^{2d} }\) multiplicative factor.
\end{remark}
%
%
%
\begin{remark}
     One advantage of theorem \ref{theorem-tensor-rank-other} (over theorem \ref{theorem-highrank-tensor-intro}) is that it involves exponentially smaller primes, even if there are many. We can lower the size of the prime in theorem \ref{theorem-highrank-tensor-intro} as follows. However, the reduction in the size of the primes involved is not significant. Let \(p_1=2\). Inductively, let \(p_j\) be the smallest prime greater than \(P_j^{P_j\varphi(P_j)}\), where \(P_j:=p_1p_2\ldots p_{j-1}\) and \(\varphi\) is the Euler totient. 
     Let \(N\) be the smallest number of the following form \(N=p_1p_2\ldots p_j\) greater than \(n^{2d(n^d+1)}\). Take a primitive \(N\)-th root of unity in theorem \ref{theorem-highrank-tensor-intro}, in place of the \(p\)-th root of unity. In the proof, recast Cheboratev's theorem with \cite[Thm. 1.4]{clmp}. 
\end{remark}
%
%
%

%
%
\bibliographystyle{splncs04}
\bibliography{references}
%




\end{document}